\def\year{2020}\relax
\documentclass[letterpaper]{article} 
\usepackage{aaai20}  
\usepackage{times}  
\usepackage{helvet} 
\usepackage{courier}  
\usepackage{amsmath,,amsfonts,amssymb}
\usepackage{array,multirow,graphicx}
\usepackage{float}
\usepackage{CJKutf8}
\usepackage{bm}
\usepackage{natbib}
\usepackage[hyphens]{url}  
\usepackage{graphicx} 
\usepackage{amsthm}
\usepackage{graphicx}  
\newcommand{\approach}[1]{approach}
\newtheorem{theorem}{Proposition}[section]

\usepackage{etoolbox}
\makeatletter
\patchcmd{\maketitle}{\@copyrightspace}{}{}{}
\makeatother
\title{Understanding Semantics from Speech Through Pre-training}
\author{
\Large \textbf{Pengwei Wang, Liangchen Wei, Yong Cao, Jinghui Xie, Yuji Cao, Zaiqing Nie} \\
Alibaba AI Labs \\
hoverwang.wpw, liangchen.wlc, yohncao.cy, jinghui.xjh, yuji.cyj, zaiqing.nzq@alibaba-inc.com
}

\begin{document}
 
\maketitle

\begin{abstract}
End-to-end Spoken Language Understanding (SLU) is proposed to infer the semantic meaning directly from audio features without intermediate text representation. 
Although the acoustic model component of an end-to-end SLU system can be pre-trained with Automatic Speech Recognition (ASR) targets, the SLU component can only learn semantic features from limited task-specific training data.
In this paper, for the first time we propose to do large-scale unsupervised pre-training for the SLU component of an end-to-end SLU system, so that the SLU component may preserve semantic features from massive unlabeled audio data.
As the output of the acoustic model component, i.e. phoneme posterior sequences, has much different characteristic from text sequences, we propose a novel pre-training model called BERT-PLM, which stands for \textbf{B}idirectional \textbf{E}ncoder \textbf{R}epresentations from\textbf{ T}ransformers through \textbf{P}ermutation \textbf{L}anguage \textbf{M}odeling.
BERT-PLM trains the SLU component on unlabeled data through a regression objective equivalent to the partial permutation language modeling objective, while leverages full bi-directional context information with BERT networks.
The experiment results show that our approach out-perform the state-of-the-art end-to-end systems with over 12.5\% error reduction. 
\end{abstract}

\section{Introduction}
Spoken Language Understanding (SLU) has attracted very much attention in recent years with the rapidly growing demand of voice interface applications and devices such as Siri, Cortana, Alexa and Google Home etc.  \citep{Bhargav:2013,Ravuri:2015,Sarikaya:2014,Tur:2011,Tur:2012,Xu:2014,Yao:2014,Siddhant:2019,Zhao:2019}. 
The conventional SLU approaches typically consist of two parts: automatic speech recognition (ASR), which convert audio into the underlying text, and natural language understanding (NLU), which takes the converted text as input \citep{Coucke:2018,Gorin:1997,Mesnil:2015}.
The major drawback of such approaches is that the NLU part suffer from the ASR errors, which set an accuracy upper bound of the entire system.
End-to-end SLU approaches, which directly infer semantic meaning from audio feature, are proposed to eliminate the above error propagation issue \citep{Chen:2018,Haghani:2018,Serdyuk:2019,Lugosch1:2019}

An end-to-end SLU system usually consists of two parts: the acoustic model component and the SLU component.
Existing end-to-end SLU approaches generally require much more training data, comparing to text based SLU, to make the acoustic model component converge.
Researchers use ASR targets, words and phonemes, to pre-train the acoustic model component in the end-to-end model, while the pre-trained acoustic model are connected to the SLU component with parameter weights either frozen or not \citep{Lugosch1:2019}.
As the acoustic model component does not capture any semantic features but the acoustic features only, acoustic pre-training may use any ASR data outside the downstream SLU task domains. 
Moreover, the acoustic model itself may also benefit from unsupervised pre-training via sequence auto-encoder \citep{Qian:2017} and noise contrastive binary classification \citep{Schneider:2019}.
However, the SLU component in the end-to-end model, which supposes to learn semantic features, still suffers from the limited task specific training data. 

Large-scale unsupervised pre-training via auto-regressive and auto-encoding objectives have been proven effective in text processing. \citep{Radford:2018,Devlin:2019,Yang:2019} 
Semantic features are preserved during the pre-training on massive text data to help the downstream text processing tasks through weight fine-tuning.
In end-to-end SLU approaches, although there is no text representations, the acoustic model output acts as an intermediate representation that holds semantics.
It's very promising to improve the end-to-end system accuracy if we can introduce semantic feature pre-training with proper objectives over massive unlabeled audio data.

In this paper, for the first time we propose a large-scale unsupervised pre-training approach for the SLU component in end-to-end SLU systems.
The acoustic model component is trained with phoneme and text based Connectionist Temporal Classification (CTC) loss or alignment loss(input sequences are labeled at frame level), following \citet{Lugosch1:2019}.
The phoneme posterior output of the acoustic model component is fed into the SLU component.
The phoneme posterior output can be regarded as a special language representation.
Pre-training the SLU component with phoneme posterior input allows the model to capture semantic features from massive unlabeled data. 
Different from text sequence, the phoneme posterior output from the acoustic model component consists of phoneme distributions over sampled time slices. 
The sampling strategy in \citep{Devlin:2019}, which is designed to handle text unigram, can hardly be implemented over the phoneme posterior output.
We employ the partial permutation language model as pre-training objective which is proposed by \citet{Yang:2019}
However, the XLNET implementation that they carry out can only leverage incomplete sequence information at each time point.
The attention masking strategy in XLNET stops it from looking at the right part of the permuted sequence, although XLNET is possible to leverage partial context from both sides of the time point in the original sequence because of permutation.
It's generally required to ``see'' the whole sequence context at each time point so that the learned representations are built up with full context information.
Moreover, the existing of a special phoneme ``SIL'', which stands for a silence time slice, corrupts the permutation language modeling objective since predicting a silence time slice does not make sense at all just like predicting a white space from an English sentence. 
Therefore, we carry out a novel pre-training model called BERT-PLM, which employs an regression objective equivalent to the partial permutation language model objective while eliminating the influence of ``SIL'' and preserving full bi-directional context information with BERT networks.

We train BERT-PLM over more than 4000 hours audio data from multiple sources.
The pre-trained model is further fine-tuned on two SLU tasks: Fluent Speech Command dataset \citep{Lugosch1:2019} and an In-House speech command dataset.
The acoustic model component weights are pre-trained and frozen during the SLU component fine-tuning. 
On all datasets, our approach significantly outperform the state-of-the-art end-to-end approaches, which do not employ unsupervised pre-training. We summarise our contributions as follows:
\begin{itemize}
\item We are the first to propose large-scale unsupervised pre-training for end-to-end SLU systems.
\item We propose a novel pre-training method BERT-PLM, which is learned through partial permutation language model objective with full context information.
\item We conduct empirical studies to show the effectiveness of our large-scale unsupervised pre-training approach.
\end{itemize}

The result of the paper is organized as follows. In section 2, we introduce related work; in section 3, we present the end-to-end model, its acoustic model component pre-training strategy and its SLU component pre-training strategy; in section 4, we show the empirical study results, while in the last section we conclude our work.

\section{Related Work}
\label{sec:related_work}
In this section, we briefly review the work that is related to this paper in two subsections: end-to-end SLU models and pre-training approaches.

\subsection{End-to-End SLU Models}
Researchers have been working on end-to-end SLU models since two years ago.
\citet{Qian:2017} employ single direction RNNs to encode input audios into  distributed latent vectors, which are further connected with semantic utterance classification layers. 
They also use a sequence auto-encoder model to pre-train the acoustic RNNs encoder to capture better distributed latent representations. 
However, the proposed end-to-end SLU model is finally used to build ensembles with the ASR+NLU models since its accuracy has an obvious gap from the ASR+NLU models.
\citet{Serdyuk:2019} reinforce the audio encoder part with multi-layer bi-directional RNNs together with a sub-sampling strategy, which makes their end-to-end model perform very close to the conventional ASR+NLU model.
\citet{Chen:2018} introduce an end-to-end model, where the acoustic model component and the SLU component are connected through the softmax probabilities over graphemes output by the acoustic model component, so that the acoustic model component can be pre-trained with ASR targets. 
The proposed end-to-end model outperform the conventional ASR+NLU model.
Pre-training acoustic model component with ASR targets is also employed in \citet{Lugosch1:2019} where the pre-training targets are phonemes and words.

\subsection{Pre-training Approaches}
Pre-training has been proven effective for neural networks as it allows the target function to be set at a better start position for further optimizing. 
Existing end-to-end models employ pre-training strategies in different ways.
\citet{Chen:2018} and \citet{Lugosch1:2019} use ASR targets for supervised pre-training of the acoustic model components.
\citet{Qian:2017} use sequence auto-encoder to do unsupervised pre-training on the acoustic model components, while \citet{Schneider:2019} employs a contrastive loss that requires distinguishing a true future audio sample from negatives.
However, pre-training is only applied to the acoustic model components in existing end-to-end SLU approaches.

Pre-training has attracted very much attention in the NLU domain.
Large-scale unsupervised pre-training approaches have been proven effective in almost all NLU tasks. \citep{peters2018deep,Devlin:2019,sun2019ernie,Radford:2018,Yang:2019}
\citet{peters2018deep} introduce ELMo, a deep bi-directional contextual word representation based on RNNs.
\citet{Radford:2018} propose GPT and GPT2, which pre-train the Transformer decoder structure via language modeling.
However, only partial contextual information is observed in ELMo, GPT and GPT2, either from left or right, at each word position.
\citet{Devlin:2019} come up with BERT, where the full contextual information from both sides of a word are used to build its contextual representations.
\citet{sun2019ernie} propose ERNIE to further refine the masked language modeling loss of BERT to leverage the entity information in a sequence.
However, the masked language modeling loss brings an placeholder token ``[MASK]'' to the vocabulary which is unseen in the downstream tasks.
\citep{Yang:2019} carry out XLNET, which employs the permutation language modeling loss to overcome the unseen token problem.
Although XLNET tries to leverage context from the both sides of a word by introducing permutation, its network implementation can only leverage incomplete context, only the left part of the word position in the permuted sequence.

Inspired by the great success of large-scale unsupervised pre-training approaches in NLU, we come up with the idea of conducting large-scale unsupervised pre-training on the SLU component of end-to-end SLU models. The approach details are presented in the following sections.

\section{Models}
\label{sec:model}

\begin{figure*}[bt!]
	\centering
	\includegraphics[scale=0.65]{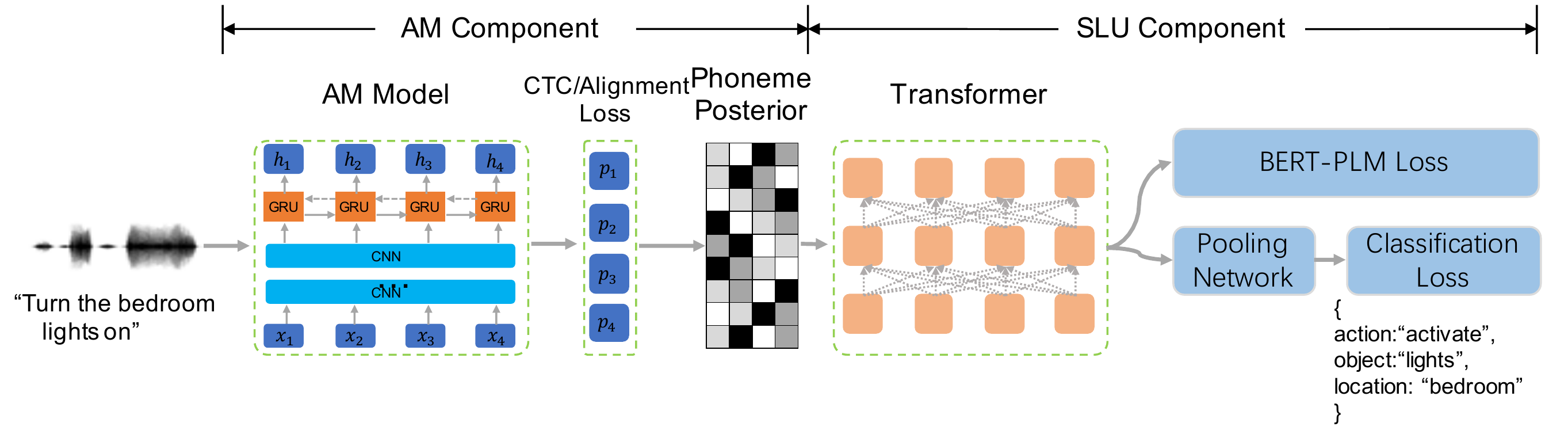}
	\caption{
		\label{fig:whole_arc} End-to-end SLU model with SLU pre-training}
\end{figure*}

The overall architecture of our end-to-end SLU model is shown in Figure~\ref{fig:whole_arc}. 
There are two major components: the acoustic model (AM) component and the SLU component, where the SLU component are connected through the posterior probabilities over phonemes output by the acoustic model component.
In the following of this section, we present the model details and the pre-training strategies.

\subsection{Acoustic Model Component}
The SLU component takes the phoneme posterior distribution output of the acoustic model component as its input.
Such implementation makes the SLU component independent from the acoustic model implementation conditioned on the phoneme posterior. 
Therefore, the SLU component is less sensitive to the choice of the acoustic model implementation as long as it provides proper phoneme posterior with comparable accuracy.
In this paper, we try two different acoustic model implementations with same SLU networks. 

\subsubsection{SincNet}
We implement an acoustic model using a SincNet layer \citep{ravanelli2018speaker,ravanelli2018interpretable} to process raw audio signals, followed by multiple convolution and bi-directional GRU layers with pooling and dropout. 
The outputs include the phoneme-level logits and the word-level logits as well.
The SincNet based acoustic model is pre-trained with frame-aligned ASR training data following \citet{Lugosch1:2019}.

\subsubsection{DFSMN}
We also employ a DFSMN\citep{zhang2018deep} acoustic model to test the robustness of the SLU component pre-training with different acoustic model implementation.
The DFSMN model is pre-trained with CTC loss using an in-house labeled ASR dataset.
In Figure~\ref{fig:whole_arc}, we illustrate the acoustic model component with the SincNet implementation only due to the limited space.

\subsection{SLU Component}
In our end-to-end model, we employ Transformer encoder network as the SLU component, since it has been proven effective in almost all NLU tasks \citep{Devlin:2019}.
Since its input is the phoneme posterior output of the acoustic model component $\bm{h}^{phoneme}$, the token embedding input of the Transformer encoder network is calculated as follows.
\begin{equation}
\label{token_embedding}
    \bm{v}^{token} = \bm{E}^{phoneme}\cdot\bm{h}^{phoneme}
\end{equation}
As $\bm{h}^{phoneme}$ represents a probability distribution, its elements are positive real numbers with summation being one.
The token embedding input can be regarded as a weighted pooling of the phoneme embedding $\bm{E}^{phoneme}$ over the phoneme vocabulary with respect to the attention weights holding in $\bm{h}^{phoneme}$.
In each Transformer block, we employ relative positional encoding, as depicted in \citet{Dai:2019}, for pre-training purpose.
The output of the last Transformer block is fed into a pooling layer where the input audio is finally converted into a vector representation for further SLU classification tasks.

\subsection{SLU Component Pre-training}
Since the acoustic model component is pre-trained with ASR targets as mentioned above.
Each phoneme posterior vector $\bm{h}^{phoneme}$ comes from a sampled audio slice over a certain period of time, typically 10$\thicksim$30 milliseconds, so that it carries very little semantic information.
Moreover, the $\bm{h}^{phoneme}$ may be rather noisy due to the audio background noise, unusual pronounced tones and in-perfectly learned acoustic model.
It requires massive task specific training data for the SLU component to capture the underlying semantics with such noisy representations.
As we can use the learned acoustic model component to generate the phoneme posterior representations for any input audio, we employ an unsupervised pre-training approach.

\subsubsection{Permutation Language Modeling}
Although a phoneme posterior sequence $\bm{H}^{phoneme}$ is different from natural language token sequences, it is still a special language representation sequence.
The unsupervised pre-training approaches for natural language token sequences can also be employed for phoneme posterior sequences as most of them are language independent.
\citet{Devlin:2019} use a masked language model target to train BERT on massive unlabeled natural language corpus. 
However the masked language model requires a special sampling strategy designed for unigram sequences. 
It's hard to to applied on phoneme posterior sequences where there is a distribution vector at each time point.
Inspired by \citet{Yang:2019}, we employ the permutation language model target proposed in their work for the SLU component pre-training.

We use the permutation language model with partial prediction targets as proposed by \citet{Yang:2019}.
\begin{equation}
\label{partial_plm}
    \log p_\theta(\mathbf{x}_{\mathbf{z}_{>c}}|\mathbf{x}_{\mathbf{z}_{\leq c}}) = \sum_{t=c+1}^{|\mathbf{z}|}{\log p_\theta(x_{z_t}|\mathbf{x}_{\mathbf{z}_{<t}})}
\end{equation}
where $\mathbf{z}$ represents a given permutation of the target sequence $\mathbf{x}$, $\mathbf{z}_{<t}$ stands for the left part of the time point $t$ in the permutation, and $c$ is the cutting point where the right part $\mathbf{z}_{>c}$ are the target subsequence.
The objective is to maximize the log-likelihood of the target subsequence conditioned on the non-target subsequence, that is
\begin{equation}
\label{partial_plm_loss}
\begin{aligned}
    &\max \mathbb{E}_{\mathbf{z}\sim Z_T}\big{[} \log p_\theta(\mathbf{x}_{\mathbf{z}_{>c}}|\mathbf{x}_{\mathbf{z}_{\leq c}})\big{]} \\
    &= \mathbb{E}_{\mathbf{z}\sim Z_T}\big{[} \sum_{t=c+1}^{|\mathbf{z}|}{\log p_\theta(x_{z_t}|\mathbf{x}_{\mathbf{z}_{<t}})} \big{]}
\end{aligned}
\end{equation}

\subsubsection{BERT-PLM}
The major challenge for the permutation language model implementation, like XLNET, is the sequential prediction issue.
When predicting the target $x_{z_t}$, the model is forbidden to look at the subsequence $\mathbf{x}_{\mathbf{z}>t}$.
As a result, the subsequence $\mathbf{x}_{\mathbf{z}>t}$ is masked when modeling $\mathbf{x}_{\mathbf{z} \leq t}$ at time point $t$ \citep{Yang:2019}.
The major drawback of such implementation is that the model can never see the full context at each time point although introducing permutation enables the model to look at the context from both sides in the original sequence. 

The root cause of this situation is the common sense of how to calculate the objective in Eq.\ref{partial_plm_loss}, which can be factorized as follows.
\begin{equation}
\label{partial_plm_loss_factors}
\begin{aligned}
    &\mathbb{E}_{\mathbf{z}\sim Z_T}\big{[} \log p_\theta(\mathbf{x}_{\mathbf{z}_{>c}}|\mathbf{x}_{\mathbf{z}_{\leq c}})\big{]} \\
    &= \frac{1}{T!} \sum_{\mathbf{z}\in Z_T}
    \sum_{t=c+1}^{|\mathbf{z}|}{\log p_\theta(x_{z_t}|\mathbf{x}_{\mathbf{z}_{<t}})} 
\end{aligned}
\end{equation}
It's natural to calculate along the summation order, in which we sample a permutation first, then try to calculate the rest summation from $c$ to $T$.
In order to calculate the summation from $c$ to $T$ in parallel, we have to choose incomplete context masking, as proposed by \citet{Yang:2019}, and sacrifice model representation capability.
We can sort out the situation by changing the calculation order of the two summations in Eq.\ref{partial_plm_loss_factors}.

\begin{theorem}
\label{equality_prop}
Given a sequence $x$ with length $T$, the partial language modeling expectation with cutting point $c\in(0,T)$ over sequence permutations $Z_T$ equals to the partial language regression expectation over the context subsequence $\hat{\mathbf{x}}$ combinations, i.e.
\begin{equation}
\label{bert_plm_loss}
\begin{aligned}
    &\mathbb{E}_{\mathbf{z}\sim Z_T}\big{[} \log p_\theta(\mathbf{x}_{\mathbf{z}_{>c}}|\mathbf{x}_{\mathbf{z}_{\leq c}})\big{]} \\
    &=\mathbb{E}_{|\overline{\mathbf{x}}|\sim [1, T-c]} \left[\log p_\theta(\overline{\mathbf{x}}|\hat{\mathbf{x}}) \right]
\end{aligned}
\end{equation}
\end{theorem}
 
\begin{proof}
To prove the above proposition, we need to change permutation to combination as the context subsequence permutations are all equivalent due to the relative position encoding.
\begin{equation} \label{proof_01}
\begin{aligned}
    &\mathbb{E}_{\mathbf{z}\sim Z_T}\big{[} \log p_\theta(\mathbf{x}_{\mathbf{z}_{>c}}|\mathbf{x}_{\mathbf{z}_{\leq c}})\big{]} \\
    &=\frac{1}{T!} \sum_{\mathbf{z}\in Z_T}
        \sum_{t=c+1}^{|\mathbf{z}|}{\log p_\theta(x_{z_t}|\mathbf{x}_{\mathbf{z}_{<t}})} \\
    &=\frac{(T-t+1)!(t-1)!}{T!} \sum_{t=c+1}^{|\mathbf{z}|}
    \sum_{i=1}^{\mathrm{C}^{T-t+1}_{T}}
    \sum_{j=1}^{T-t+1}
    \log p(x_j|\mathbf{\hat{x}}_i)
\end{aligned}
\end{equation}
where $\mathbf{\hat{x}}_i$ stands for one possible combination of $t-1$ elements from the sequence. Let $\overline{\mathbf{x}}$ denotes the combination of all possible target elements $x_j$, then
\begin{equation} \label{proof_02}
\begin{aligned}
    &\mathbb{E}_{\mathbf{z}\sim Z_T}\big{[} \log p_\theta(\mathbf{x}_{\mathbf{z}_{>c}}|\mathbf{x}_{\mathbf{z}_{\leq c}})\big{]} \\
    &=\frac{1}{\mathrm{C}^{T-t+1}_{T}} \sum_{t=c+1}^{|\mathbf{z}|}
    \sum_{i=1}^{\mathrm{C}^{T-t+1}_{T}}
     \log p(\overline{\mathbf{x}}|\mathbf{\hat{x}}_i) \\
    &=\mathbb{E}_{|\overline{\mathbf{x}}|\sim [1, T-c]} \left[\log p_\theta(\overline{\mathbf{x}}|\hat{\mathbf{x}}) \right]
\end{aligned}
\end{equation}
where $\hat{\mathbf{x}}$ denotes one possible context subsequence combination with length $t\in[c,T-1]$.
\end{proof} 

Proposition \ref{equality_prop} provides a new sampling strategy without permutation. Firstly, we sample a context subsequence $\hat{\mathbf{x}}$, a combination of the sequence elements with length in $[c, T-1]$. Secondly, we predict all the rest elements independently, which forms the combination of the target subsequence $\overline{\mathbf{x}}$.
The above sampling strategy provides a stable context subsequence and a corresponding stable target subsequence given a specific input sequence.
As a result, we may employ the full bi-directional Transformer encoder networks to learn representations for context subsequences.
We call the corresponding model BERT-PLM, an alternative from BERT.
However, the difference between BERT-PLM and BERT is summaized as follows.
\begin{itemize}
\item The component $\log p_\theta(\overline{\mathbf{x}}|\hat{\mathbf{x}})$ in Eq.\ref{bert_plm_loss} is different from the masked language model objective function in BERT. $\hat{\mathbf{x}}$ in BERT stands for a corrupted sequence with manually added token ``[MASK]'', while in BERT-PLM it stands for a sample of context subsequence with no manually added token.
\item BERT-PLM employs relative positional encoding so that the permutations of one same context subsequence combination are all equivalent. 
\end{itemize}

When implementing BERT-PLM, the sampling of $\hat{\mathbf{x}}$ is done by simply masking the values from the target subsequence $\overline{\mathbf{x}}$ as shown in Figure~\ref{fig:bertpl}.
However, we do need the position information when predicting each individual $\overline{x}_{z_t}$, so we replace the input token embedding at ${z_t}$ with a learnable vector $w$, while fetching the output at ${z_t}$ from the last layer as the context representations for prediction. 
Since $\overline{\mathbf{x}}$ is masked constantly, the context representations at ${z_t}$ will never see the ground-truth $\overline{x}_{z_t}$, but will keep the position information according to the relative positional encoding.
Therefore, we can calculate the context representation with position information in one stream, which will save the computation resource comparing to the XLNET's two-stream implementation.
Moreover, as we employ full bi-directional Transformer encoder networks, BERT-PLM does not suffer from the incomplete context information problem of the XLNET implementation for the permutation language model.

\begin{figure}
    \centering
    \includegraphics[scale=0.8]{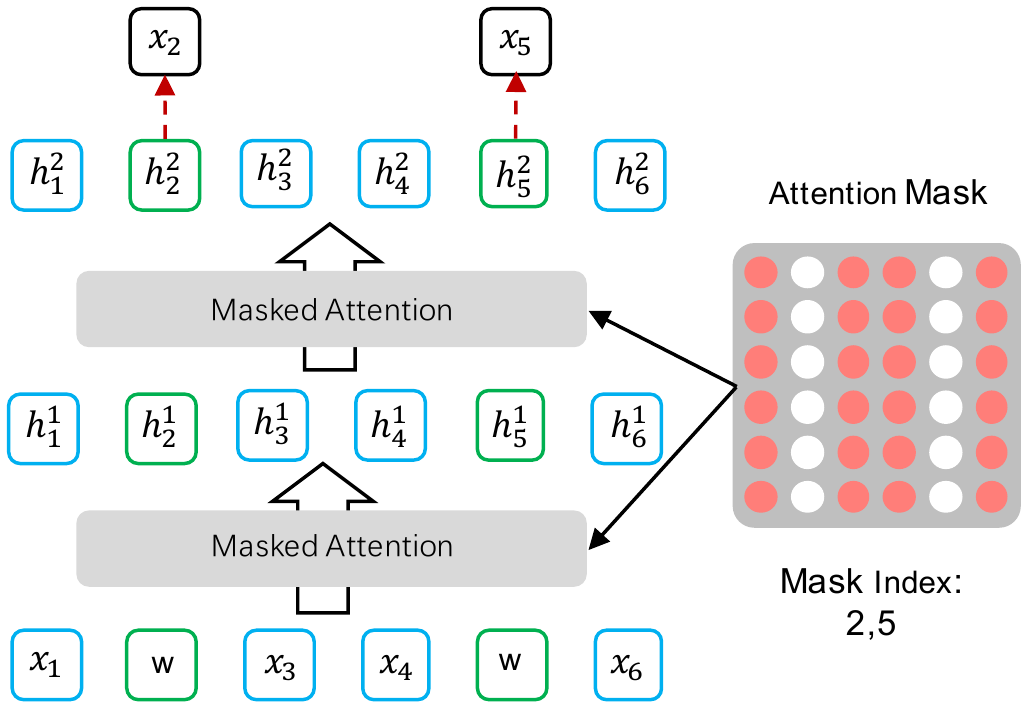}
    \caption{BERT-PLM: Masked Dot-Product Attention for Permutation Context Subsequence $\hat{\mathbf{x}}$ }
    \label{fig:bertpl}
\end{figure}

\subsubsection{Handling ``SIL''}
Different from text, the phoneme posterior sequence element $x_{z_t}$ is a distribution over the phoneme vocabulary.
When predicting $x_{z_t}$ through permutation language modeling, we still employ cross entropy loss to approximate the ground-truth distribution with the predicted one.
There is a special phoneme ``SIL'', which stands for a silence time slice.
There usually be multiple ``SIL'' between meaningful phonemes due to the interval silence during pronunciation.
On the other hand, the probability of ``SIL'' in a posterior distribution also reflects the pronunciation volume.
For example, a distribution of 0.5 ``SIL'' and 0.5 ``S'' means the ``S'' is pronounced with half of the max volume.
Predicting a major ``SIL'' time slice does not make sense at all, just like predicting spaces from an English sentence.
However, we cannot simple remove ``SIL'' since it models volume.
The existing of ``SIL'' makes the XLNET implementation fail to apply since the existing of ``SIL'' phonemes in the target subsequence corrupts the language modeling objective.
However, ``SIL'' is not a problem for BERT-PLM.
The BERT-PLM sampling strategy allows us to exclude the major ``SIL'' time slices from the target subsequence $\overline{x}_{z_t}$.

\subsubsection{Unsupervised Pre-training and Finetuning}
The SLU component is pre-trained on the phoneme posterior output of the pre-trained acoustic model component from unlabeled audio corpus via the permutation language model objective depicted above.
The pre-trained parameter weights are further finetuned in the task specific end-to-end model.
The acoustic model component parameters are frozen during the finetuning.
We premutation language modeling loss for pre-training is still used during the finetuning process together with the target loss.
Therefore, the target subsequence that is masked from the phoneme posterior sequence is hidden during the finetuning, which can be regarded as input drop-out regularization.
For the classification tasks, we further employ a pooling network that proposed by \citet{Lin:2017} after the BERT-PLM networks. Instead of taking the embedding of the start flag as the output vector, the pooling network use the hidden matrix to learn the representation via a single-head attention operation, as shown in Figure \ref{fig:pooling}. The $Q$ in the figure is a trainable vector and randomly initialized.

\begin{figure}
    \centering
    \includegraphics[scale=0.8]{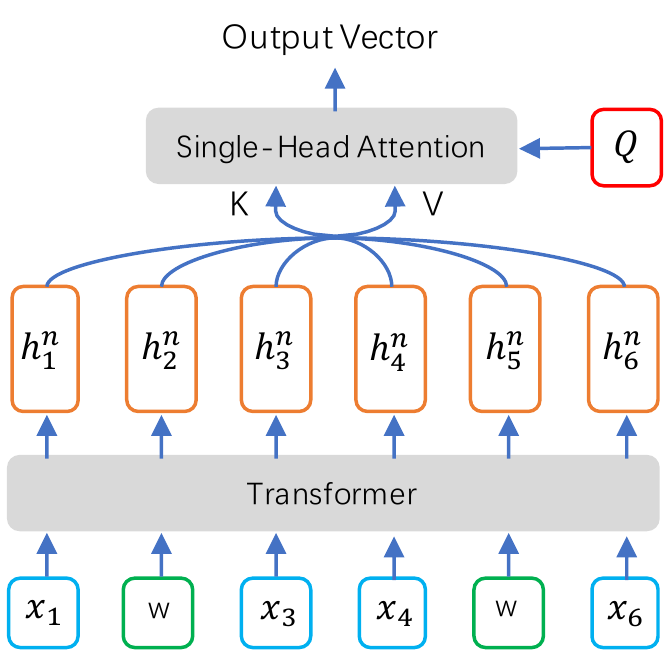}
    \caption{Pooling Network}
    \label{fig:pooling}
\end{figure}

\section{Experiment}
\label{sec:experiment}

In this section, we evaluate the effectiveness of our unsupervised pre-training approach by comparing to various baseline approaches.
We also conduct experiments on different model setting to see the influence of pre-training.
We choose two datasets on different langauges: English and Mandarin, to test the robustness of our approach.

\subsection{Data Collection}
The dataset used in this work consists of three parts, labeled ASR dataset for training the acoustic model component, large scale audio clips for unsupervised phoneme-level pretraining of the SLU component and the end-to-end labeled SLU dataset.
We utilize the Libri Speech\citep{panayotov2015librispeech} dataset and a in house labeled ASR dataset for two implementations of the acoustic model component, respectively, as shown in Table.\ref{tab:asr_dataset}.

\begin{table}
\caption{Datasets for AM training }
\label{tab:asr_dataset}       
\centering
\begin{tabular}{lcc}
\hline\noalign{\smallskip}
AM Impl & Dataset & Hours  \\
\noalign{\smallskip}\hline\noalign{\smallskip}
SyncNet & Libri Speech & 1000 \\
DFSMN & In-House ASR Dataset & 20000 \\
\noalign{\smallskip}\hline
\end{tabular}
\end{table}

For SLU component pre-training, we use 4000+ hours of audio clips from several public datasets, including Common Voice, Libri Speech,  VoxCeleb1\&2\citep{nagrani2017voxceleb} and the Free Spoken Digit Dataset. We only use the audio clips to pretrain the SLU component, which later will be used for fine-tuning. 
Statistics of the dataset for pre-training is shown in Table.\ref{tab:pretrain_dataset}

As for the SLU datasets, we use Fluent Speech Commands \citep{Lugosch1:2019} and a self-build In-House Speech Commands dataset. The Fluent Speech Commands dataset is public, which contains approximately 30k audio-command pairs in English. The In-House dataset contains about 20k audio-command pairs of near field annotated data collected from a diverse set of more than 1000 speakers. See Table.\ref{tab:slu_dataset} for details, where \#C denotes the number of speech command type.

\begin{table}
\caption{The Statistic of Pretrain-Datasets}
\label{tab:pretrain_dataset}       
\centering
\begin{tabular}{lcc}
\hline\noalign{\smallskip}
Dataset & Hours  \\
\noalign{\smallskip}\hline\noalign{\smallskip}
Vox-Celeb & 2,000h  \\
Libri Speech & 1,000h \\
Common Voice & 1,087h \\
Free Spoken Digit & - \\
\noalign{\smallskip}\hline
\end{tabular}
\end{table}

\begin{table}
\caption{The Statistic of SLU-Datasets}
\label{tab:slu_dataset}       
\centering
\begin{tabular}{lcccc}
\hline\noalign{\smallskip}
Dataset & \#C & Train & Test & Valid  \\
\noalign{\smallskip}\hline\noalign{\smallskip}
Fluent Speech Command & 31 & 23.1k & 3.8k & 3.1k\\
In-House & 16 & 20k & 2k & -\\
\noalign{\smallskip}\hline
\end{tabular}
\end{table}

\subsection{Baselines}
\textbf{ASR + NLU pipeline}. To verify the effectiveness of the end-to-end models on Fluent Speech Command dataset, we implement an ASR + NLU pipeline as one of our baselines. As we train the acoustic model with phoneme-level and word-level target loss together, naturally we have an ASR model that is trained with the same labeled ASR dataset. As for the NLU model, we still adopt the Transformer architecture by applying a lookup table over the word sequences. We denote this approach by ASR+NLU in Table.\ref{tab:fsc} and Table.\ref{tab:ih}. 

\noindent\textbf{Non-Pretraining}. The end-to-end approach proposed in this paper without the SLU component being pretrained with audio data.

\noindent\textbf{SOTA}. The state-of-the-art model on Fluent Speech Command dataset, which is proposed by \citep{Lugosch1:2019}. It is an end-to-end model whose ASR module is pre-trained with Libri Speech dataset. The differences between the Non-Pretraining approach and the SOTA approach lies in two ways: 1) Classifier. The SOTA approch adopt a bidirectional-RNN encoder followed by max-pooling to squash the sequence of outputs from the recurrent layer into a single vector of logits corresponding to the different slot values. While we adopt a Transformer architecture followed by a full connection layer for classification. 2) Inputs of the classifier. The SOTA approach takes the word-level hidden representation as input while our approach takes the phoneme posterior as inputs.



\subsection{Experiment Settings}
\textbf{Hyper-parameters}. The BERT-PLM are implemented with 4 layers of multi-head self-attention, the phoneme embedding size is set to 576, intermediate size set to 1600 and number of heads of self-attention set to 8. The learning rate is set to $3*10^{-5}$, the dropout rate is set to 0.1. Typically a sampled audio slice is 30 milliseconds, we set the max input sequence length to be 320 which is long enough to hold audio clips no more than 10 seconds. The pretrained model converges within 10 epoches of training.

\noindent\textbf{Environment}. All models are implemented with Tensorflow 1.12\citep{abadi2016tensorflow}. We conduct the pre-training experiment on 8 Nvidia-Tesla P-100 GPUs for approximately 7 hours. 

\subsection{Experiment Results}
We compare different approaches on two SLU datasets mentioned above. 
As can be seen from Table.\ref{tab:fsc}, given the same ASR dataset and fine-tuning dataset, end-to-end approaches, including the state-of-the-art and Non-Pretraining, significantly outperform the ASR + NLU pipeline by a large margin, meaning that the ASR component suffers from more information loss compared to the acoustic model component as the downstream Transformer architecture for Non-Pretraining and ASR+NLU are the same. 
In other words, the end-to-end models can efficiently alleviate the ASR error problem in ASR+NLU pipeline. 

\begin{table} [h]
\caption{Comparison of error rate between different approaches on the Fluent Speech Command dataset}
\label{tab:fsc}       
\centering
\begin{tabular}{lrr}
\hline\noalign{\smallskip}
Model & Error Rate & Delta  \\
\noalign{\smallskip}\hline\noalign{\smallskip}
SOTA & $1.2\%$ & $-$ \\
ASR+NLU & $9.89\%$ & $-724.17\%$ \\
Non-Pretraining & $1.95\%$ & $-62.50\%$ \\
BERT-PLM & $\mathbf{1.05\%}$ & $\mathbf{12.5\%}$ \\
\noalign{\smallskip}\hline
\end{tabular}
\end{table}

The state-of-the-art approach outperforms the Non-Pretraining approach because the state-of-the-art approach utilize the word-level hidden representation which carries more semantic information in comparison to the Non-Pretraining approach which utilize the phoneme posterior distribution without pre-training.
However, BERT-PLM significantly outperform the state-of-the-art approach as well as the Non-Pretraining approach, indicating that large-scale unsupervised pre-training can help the SLU component capture more semantic information. 
On the other hand, our claim in previous section that the SLU component is independent from the acoustic model implementation conditioned on the phoneme posterior has been proved to be true. 
From the experiment results, we observe that the proposed end-to-end approach benefits from SLU component pre-training regardless of implementation of the acoustic model as we adopt different acoustic models for different task while keeping the SLU component the same.

To be noted that, BERT-PLM yield the highest performance which reduces the error rate by 12.5\% over the state-of-the-art approach on Fluent Speech Command dataset. 
Due to the fact that the existence of "SIL" brings trouble to our implementation of XLNET, we omit the result of XLNET in our experiment result.


\begin{table} [h]
\caption{Comparison of F1 score between different approaches on the In-House Speech Command dataset}
\label{tab:ih}       
\centering
\begin{tabular}{lcc}
\hline\noalign{\smallskip}
Model & Macro-F1 & Micro-F1\\
\noalign{\smallskip}\hline\noalign{\smallskip}
Non-Pretraining & $78.30\%$ & $89.38\%$ \\
BERT-PLM & $\mathbf{79.36\%}$ & $\mathbf{90.45\%}$ \\
\noalign{\smallskip}\hline
\end{tabular}
\end{table}

\begin{table} [h]
\caption{The results of different percentage of the number of mask}
\label{tab:var_mask}       
\centering
\begin{tabular}{lcc}
\hline\noalign{\smallskip}
Model & Error Rate & Delta\\
\noalign{\smallskip}\hline\noalign{\smallskip}
SOTA & $1.20\%$ & -\\
BERT-PLM(5\%) & $1.11\%$  & $7.5\%$ \\
BERT-PLM(10\%) & $1.11\%$ & $7.5\%$ \\
BERT-PLM(15\%) & $\mathbf{1.05\%}$ & $\mathbf{12.5\%}$ \\
BERT-PLM(20\%) & $1.13\%$ & $5.8\%$ \\
\noalign{\smallskip}\hline
\end{tabular}
\end{table}

\begin{figure}[h]
	\centering
	\includegraphics[scale=0.53]{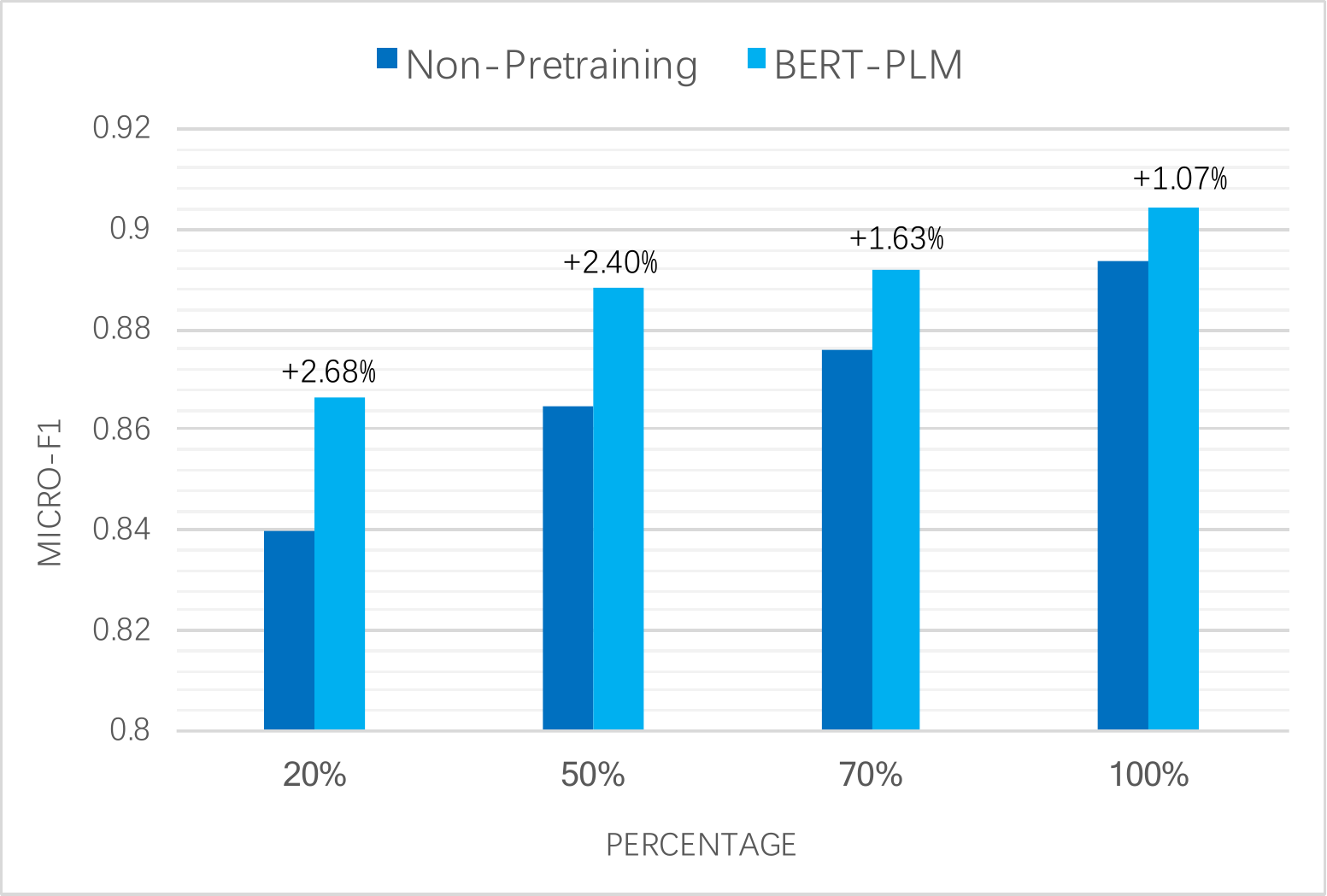}
	\caption{
		\label{fig:percentage} The F1 score of BERT-PLM Pre-training on variant percentage of in-house training data.}
\end{figure}

Next, we show the results on different percentage of the SLU training data. We use the in-house data to show some insights. We select the 20\%, 50\%, 70\% and 100\% percent out of the whole in-house dataset as training data for fine-tuning. The results are shown in Figure.\ref{fig:percentage}. From the results, we can see that the benefits from pre-training is decreased as the training data increases. This also illustrate that the semantic information can be directly learned from the labeled training data if the number of the training data is large.

We also select different maximum percentage of frames to be masked. Separately, up to 5\%, 10\%, 15\%, 20\% of the total non-sil frames of each instance are selected as target frames during pre-training. The result of fine-tuning based on the different pre-trained models are shown in Table.\ref{tab:var_mask}. We observe the best performance when we set the mask percentage upper bound to $15\%$. When the mask percentage upper bound is set to $5\%$ and $10\%$, the pre-training task is too easy for the model so that it may over-fit, but we still observe performance gain from pre-training in comparison to the state-of-the-art approach. When the mask percentage upper bound is set to $20\%$, the pre-training objective is hard to learn so that less performance gain is observed.




\section{Conclusion}
End-to-end SLU approaches have attracted much attention since the semantic features are directly learned from audio features without intermediate text representation. 
In this paper, we propose to do large-scale unsupervised pre-training for the SLU component of an end-to-end SLU system for the first time.
As a result, the SLU component may preserve semantic features from massive unlabeled data.
We employ an regression objective for pre-training which is equivalent to the partial permutation language modeling objective.
To handle the characteristic of the phoneme posterior outputs from the acoustic model component, we propose BERT-PLM, a novel pre-training model.
BERT-PLM leverages full bi-directional context information with BERT networks.
The experiment results show the effectiveness of our approach, which out-perform the state-of-the-art end-to-end systems with over 12.5\% error reduction.

\bibliography{aaai}
\bibliographystyle{aaai}

\end{document}